\documentclass{amsart}
\usepackage{amsfonts}
\usepackage{amsmath,amscd,lscape}
\usepackage{amsthm}
\usepackage{amssymb}
\usepackage{latexsym}
\usepackage{mathrsfs}
\usepackage{tikz}

\usepackage[margin=3cm]{geometry}
\usepackage{MnSymbol} 

\newcommand{\be}{\begin{equation}}
\newcommand{\ee}{\end{equation}}
\newcommand{\ba}{\begin{eqnarray}}
\newcommand{\ea}{\end{eqnarray}}
\newcommand{\baa}{\begin{eqnarray*}}
\newcommand{\eaa}{\end{eqnarray*}}
\newcommand{\bb}{\begin{thebibliography}}
\newcommand{\eb}{\end{thebibliography}}

\newcommand{\bi}[1]{\bibitem{#1}}
\newcommand{\lab}[1]{\label{#1}}
\newcommand{\re}[1]{(\ref{#1})}

\newcounter{my}
\newcommand{\he}%
   {\stepcounter{equation}\setcounter{my}%
   {\value{equation}}\setcounter{equation}0%
   }%
\newcommand{\she}%
   {\setcounter{equation}{\value{my}}%
    }%

\renewcommand\t{\tilde}

\newcommand\boX{\textsf{X}}
\newcommand\boW{\textsf{W}}
\newcommand\boZ{\textsf{Z}}
\newcommand\boH{\textsf{H}}

\newcommand\oZ{\overline{Z}}

\newcommand\vphi{\varphi}

\newtheorem{pr}{Proposition}

\newtheorem{theorem}{Theorem}[section]

\newtheorem{lemma}[theorem]{Lemma}
\newtheorem{definition}[theorem]{Definition}
\theoremstyle{definition}

\newtheorem{remark}[theorem]{Remark}

\numberwithin{equation}{section}

\begin{document}

\title{The Heun-Askey-Wilson algebra and  the Heun operator of Askey-Wilson type}

\author{Pascal Baseilhac}
\author{Satoshi Tsujimoto}
\author{Luc Vinet}
\author{Alexei Zhedanov}

\address{Institut Denis-Poisson CNRS/UMR 7013 - Universit\'e de Tours - Universit\'e d'Orl\'eans Parc de Grammont,
37200 Tours, FRANCE}

\address{Department of Applied Mathematics and Physics, Graduate School of Informatics, Kyoto University,
Yoshida-Honmachi, Kyoto, Japan 606--8501}

\address{Centre de recherches \\ math\'ematiques,
Universit\'e de Montr\'eal, P.O. Box 6128, Centre-ville Station,
Montr\'eal (Qu\'ebec), H3C 3J7}

\address{School of Mathematics, Renmin University of China, Beijing 100872,CHINA}

\begin{abstract}
The Heun-Askey-Wilson algebra is introduced through generators $\{\boX,\boW\}$ and relations. These relations can be understood as an extension of the usual Askey-Wilson ones. A central element is given, and a canonical form of the  Heun-Askey-Wilson algebra is presented. A homomorphism from the Heun-Askey-Wilson algebra to the Askey-Wilson one is identified.  On the vector space of the polynomials in the variable $x=z+z^{-1}$, the Heun operator of Askey-Wilson type realizing $\boW$ can be characterized as the most general second order $q$-difference operator in the variable $z$ that maps polynomials of degree $n$ in $x=z+z^{-1}$ into polynomials of degree $n+1$. 
\end{abstract}



\maketitle

\section{Introduction}
\setcounter{equation}{0}

We here introduce a $q$-analog ${\mathcal W}$ of the Heun operator that has the property of acting tridiagonally on Askey-Wilson polynomials. We moreover define in terms of generators and relations, the extension of the Askey-Wilson algebra that results when ${\mathcal W}$ is combined with either one of the Askey-Wilson bispectral operators. This generalized structure will be referred to as the Heun-Askey-Wilson (HAW) algebra.

This paper is part of a recent set of articles dedicated to the Heun equation and the construction of its difference analogs on various lattices. This series of studies is rooted in the observation made in \cite{GVZ_band} that a so-called algebraic Heun operator (AHO) can be associated to any bispectral problem. This AHO is simply the generic bilinear combination of the two operators defining the bispectral setting. The motivation for the introduction in \cite{GVZ_band} of the notion of AHO was to identify  these constructs as the commuting operators in problems of the time and band limiting type. It is well known that the ($q$-)hypergeometric polynomials of the Askey scheme \cite{KLS}
all present bispectral instances whose properties are encoded in quadratic algebras \cite{GLZ_Annals} generated by the corresponding bispectral operators. These algebras are typically referred to by the name of the family of polynomials to which they are associated. It thus follows that an AHO can be attached to every entry of the Askey tableaux. Moreover the construction of the corresponding AHO amounts to the most general tridiagonalization \cite{IK1,IK2, GIVZ} of one of the defining operator of the polynomials. The reason for the AHO nomenclature rests with the fact that the tridiagonalization of the hypergeometric operator which provides the AHO connected to the Jacobi polynomials, precisely corresponds to the standard Heun operator \cite{GVZ_Heun}. The Heun operator is also determined by the fact that it is the most general second order differential operator that maps polynomials of degree $n$ onto polynomials of degree $n+1$. It is striking that this condition and the AHO construction based on Jacobi polynomials yield equivalent characterizations. This has useful consequences that have been spelled out in \cite{GVZ_Heun}. Even more striking is the fact that this pattern extends to the difference operator realm.

The first observation in this direction has been made in \cite{VZ_HH} where it has been shown that the most general second order difference operator on the uniform lattice that increases by one the degree of polynomials defined on this grid, coincides in fact with the algebraic Heun operator associated to the Hahn polynomials. This has led to what was called the Heun operator of the Hahn type which provides a difference analog of the standard Heun equation. $Q$-analogs of the Heun operator were subsequently studied in \cite{BVZ} through a parallel analysis by replacing the uniform by the exponential lattice. The most general second order $q$-difference operator with the degree raising property was called the big $q$-Heun operator and shown to be identical to the AHO connected to the big $q$-Jacobi polynomials. A special case  readily seen to be the AHO of the little $q$-Jacobi polynomials has been called the little $q$-Heun operator. The big and little $q$-Heun operators were recognized to be two of the $q$-Heun operators introduced by Takemura \cite{Takemura} as degenerations of the Ruijenaars-van Diejen operators \cite{Ruijsenaars, Diejen} thereby adding to the understanding of these topics. We now follow up with possibly the most important case: the $q$-analog of the Heun equation on the Askey grid.

The tridiagonalization method has been used to inform the theory of higher polynomials (in the Askey scheme) from the lower ones \cite{IK1, IK2, GIVZ}. The basic idea is to built the defining operator of the higher polynomials as special bilinear combinations of the bispectral operators of the lower polynomials. This splits the parameters of the higher polynomials in two classes: on the one hand, the set of parameters pertaining to the lower polynomials and on the other, the parameters entering the tridiagonalization. From an algebraic perspective, such special tridiagonalizations effect embeddings of the quadratic algebra associated to the higher polynomials into the algebra of the lower system. This is done by replacing the defining operator of the lower polynomials by the tridiagonalized one. In the q-sector, one has the Askey-Wilson (AW) algebra \cite{Zh_AW} at the highest level. Obviously the general tridiagonalization that yields AHOs takes one outside the framework of orthogonal polynomials. It is of interest to provide an algebraic picture of this situation and to unravel the structure that extends the AW algebra to an HAW algebra. This will also be done here.

The outline of this paper is as follows. In Section 2, we introduce the Heun-Askey-Wilson (HAW) algebra. This algebra generalizes the Askey-Wilson algebra, whose defining relations can  be recovered for certain choices of the structure constants. A central element is constructed.  Also, for generic values of the deformation parameter $q$, a canonical presentation of the HAW algebra is exhibited. In Section 3, we study the relation between the HAW and the Askey-Wilson algebras. An explicit homomorphism from the  HAW algebra to the Askey-Wilson one is given,  inducing a new presentation that is closely related with the ``${\mathbb Z}_3$ symmetric'' form of the AW algebra. In Section 4 we construct the Heun-AW operator ${\mathcal W}$ as the most general $q$-derivative operator of second order in the variable $z$ that maps polynomials of order $n$ in the variable $x=z+z^{-1}$ into polynomials of degree $n+1$ in the same variable $x$. We note that the AW operator is a special case of ${\mathcal W}$. We also show that ${\mathcal W}$ is the same as the algebraic Heun operator associated to the Askey-Wilson polynomials. That ${\mathcal W}$ acts tridiagonally on AW polynomials is seen to follow. We also describe bispectrality properties of the Heun-AW operator ${\mathcal W}$. Concluding remarks in Section 5 will bring the paper to a close.\vspace{2mm}

{\bf Notations:} The parameter $q$ is assumed not to be a root of unity and is different than $1$ except when explicitly specified. We write $[X,Y]_q = qXY - q^{-1}YX$, the commutator $[X,Y]=[X,Y]_{q=1}$ and anticommutator $\{X,Y\}=XY+YX$. The identity element is $\mathcal{I}$.

\section{The Heun-Askey-Wilson algebra}
In this section, the Heun-Askey-Wilson algebra is introduced and some of its properties are exhibited. This algebra can be seen as a non-trivial deformation of the Askey-Wilson algebra. A central element is constructed, which generalizes the central element of the AW algebra. Also, the `canonial' form of the HAW algebra is described.
\begin{definition} The Heun-Askey-Wilson algebra  HAW  is generated by $\boX$ and $\boW$ subject to the following two relations:
\ba
&&[\boX,[\boX,\boW]]= \rho \boX\boW\boX + e_1 \boX^3+  b_1 \boX^2 + b_2 \{\boX,\boW\} + b_3 \boX + b_4 \boW + b_5 \mathcal{I} \ , \label{hawi1} \\
&&[\boW[\boW,\boX]]= \rho \boW\boX\boW + e_2 \boX^3 + e_3 \boX\boW\boX + e_4 \boX^2 + b'_1\{\boX,\boW\} + b_2 \boW^2 + b'_3 \boW + b_6 \boX + b_7 \mathcal{I} \lab{H_XW} \ , \qquad \label{hawi2}
\ea
where $\rho=q^2+q^{-2}-2$, $\{e_1,e_2,e_4\}$, $\{b_i\}$  are generic scalars and
\be
e_3 = e_1(\rho+3), \quad \; b_1'=b_1 + e_1 b_2,  \quad \; b_3' = b_3 + e_1 b_4 \lab{JI_conds}. 
\ee
\end{definition}
\begin{remark} For the specialization $e_1=e_2=e_3=e_4=0$, the defining relations of the Heun-Askey-Wilson algebra  reduce to those of the Askey-Wilson algebra \cite{Zh_AW}. 
\end{remark}

Let us make some comments on the restrictions  (\ref{JI_conds})  of the structure constants of the HAW algebra. To this end,  introduce the commutator 
\be
\boZ = [\boX,\boW]\ . \lab{Z_WX} 
\ee
Then, relations  \re{hawi1}-\re{hawi2} together with \re{Z_WX} constitute an algebra with cubic commutation relations
\be
[\boX,\boZ] = \rho \boX\boW\boX + e_1 \boX^3+  b_1 \boX^2 + b_2 \{\boX,\boW\} + b_3 \boX + b_4 \boW + b_5 \mathcal{I} \lab{XZ_com} \ee
and
\be
[\boZ,\boW]=\rho \boW\boX\boW + e_2 \boX^3 + e_3 \boX\boW\boX + e_4 \boX^2 + b'_1\{\boX,\boW\} + b_2 \boW^2 + b'_3 \boW + b_6 \boX + b_7 \mathcal{I}\ . \lab{ZW_com} \ee

Consider the Jacobi identity
\be
[[\boX,\boZ],\boW] + [[\boZ,\boW],\boX] + [[\boW,\boX],\boZ] = [[\boX,\boZ],\boW] + [[\boZ,\boW],\boX] =0 \lab{Jac_id} \ .
\ee
We have from \re{XZ_com},  \re{ZW_com}:
\ba
\big[\big[\boX,\boZ\big],\boW\big]&=& \rho \boZ\boW\boX + \rho \boX\boW\boZ + e_1 \{\boX^2,\boZ\} + e_1 \boX\boZ\boX + b_1 \{\boX,\boZ\} + b_2 \{\boW,\boZ\}+ b_3 \boZ \lab{XZW_1}\ ,\\
\big[\big[\boZ,\boW\big],\boX\big] &=& -\rho \boZ\boX\boW - \rho \boW\boX\boZ - e_3 \boX\boZ\boX - b_1'\{\boX,\boZ\} - b_2 \{\boW,\boZ\} -b_3'\boZ \lab{ZWX_1}\ .
 \ea
Hence
\be
[[\boX,\boZ],\boW]+[[\boZ,\boW],\boX]= e_1 \{\boX^2,\boZ\} + (e_1 -e_3) \boX\boZ\boX + (b_1-b_1')\{\boX,\boZ\} + (b_3-b_3') \boZ \lab{cm_sum_1}\ . \ee
According to the Jacobi identity \re{Jac_id}, this expression should identically vanish.  If one notices that
\be
\{\boX^2,\boZ\} =[\boX,[\boX,\boZ]]+ 2 \boX\boZ\boX = (\rho+2) \boX\boZ\boX + b_2 \{\boX,\boZ\} + b_4 \boZ \lab{X2Z} 
\ee
then  \re{cm_sum_1} simplifies to
\be
[[\boX,\boZ],\boW]+[[\boZ,\boW],\boX]= \left(e_1(\rho+2) +e_1-e_3 \right) \boX\boZ\boX + (b_1-b_1' + e_1 b_2)\{\boX,\boZ\} + (b_3-b_3'+e_1 b_4) \boZ\ . \lab{cm_sum_2}
 \ee
Thus, the Jacobi identity \re{Jac_id} implies the relations (\ref{JI_conds}) between some of the scalars.

\subsection{A characterization of the HAW algebra}
There is an interesting characterization of the HAW algebra in  terms of polynomials in noncommutative variables. Consider the r.h.s. of the defining relations   \re{hawi1}-\re{hawi2} of the HAW algebra.
\begin{pr}\lab{pr12}
The r.h.s. of \re{hawi1} is the most general  polynomial of total degree 3 in two noncommuting variables $\boX$ and $\boW$ such that

(i) the degree of $\boW$ is at most one;

(ii) all monomials in $\boX$ and $\boW$ appear through a linear combination of all possible symmetrized forms.
\end{pr}
Explicitly, by $(ii)$ any monomial of degree 2 in $\boX$ and degree 1 in $\boW$ should appear as a linear combination of the two symmetric terms: $\boX\boW\boX$ and $\boX^2\boW+\boW\boX^2$ (these are the only possible symmetrized combinations of these operators of the prescribed degrees).

\begin{pr} \lab{pr22}
The r.h.s. of \re{hawi2} is the most general  polynomial of total degree 3 in two noncommuting variables $\boX$ and $\boW$ such that

(i) the degree of $\boW$ is at most two;

(ii) all monomials in $\boX$ and $\boW$ appear through a linear combination of all possible symmetrized forms. 
\end{pr}

{\it Proof}. Consider all possible symmetrized  monomials of total degree three or less in the variables $\boX$ and $\boW$ and
of degree not greater than one in  $\boW$. These terms are (in descending order):
\be
\boX\boW\boX, \: \{\boX^2, \boW\}, \: \boX^3, \: \{\boW, \boX\},  \: \boX^2,  \: \boW, \: \boX, \mathcal{I} \ .
\ee
We can write down the double commutator as a linear combination of these terms for certain scalar parameters $\{d_i\}$.
\be
[\boX,[\boX,\boW]]= d_1 \boX\boW\boX + d_2 \{\boX^2, \boW\} +d_3 \boX^3 + d_4  \{\boW, \boX\} +  d_5 \boX^2 + d_6  \boW + d_7 \boX + d_8 \mathcal{I} \ . \lab{dbx_gen} \ee
However, the term $\{\boX^2,\boW\}$ one the r.h.s. of \re{dbx_gen} can be expressed as
\be
\{\boX^2, \boW\} = [\boX,[\boX,\boW]] + 2 \boX\boW\boX \lab{anti_comm} \ee
and hence we can eliminate this term from the r.h.s. of \re{dbx_gen}. The remaining terms are exactly those which appear in \re{hawi1}. This proves Proposition \ref{pr12}.

Quite similarly, all possible symmetrized  monomials of total degree less or equal 3 in variables $\boW$ and $\boX$ and less or equal two in the variable $\boW$ are
\be
\boW\boX\boW, \: \{\boW^2,\boX\}, \: \boX\boW\boX, \:  \boX^3, \: \boW^2, \{\boW,\boX\}, \: \boX^2, \: \boW, \: \boX, \: \mathcal{I}
\ee
and we thus arrive at the relation \re{hawi2}. This proves Proposition  \ref{pr22}.\vspace{1mm}

We thus see that the HAW algebra possesses a certain universality property.

\subsection{Central element of the HAW algebra} 
For the Askey-Wilson algebra, there is a central element that plays an important role in the representation theory of this algebra. It turns out that a central element can similarly be constructed for the
Heun-Askey-Wilson algebra. The proof of the following lemma is rather long and technical\footnote{Contrary to the case of the Askey-Wilson algebra, the presence of the monomial $\boX\boW\boX$ in the r.h.s. of the defining relation (\ref{hawi2}) imply the existence of non-trivial relations involving terms of degree one and two in $\boW$.} otherwise and we shall omit the details. Let us denote:
\ba
\boH_1=[\boW,\boX]_q \quad \mbox{and} \quad \tilde{\boH}_1=[\boX,\boW]_q. \label{H}
\ea
\begin{lemma} The element
\ba
\Omega^{HAW}&=& A \boX\boW \boH_1 +B \boX^2 +C\boW^2 +D \boH_1^2  + E\boX + F\boW +  G\boH_1 + H\tilde{\boH}_1 \lab{OmHAW}\\ 
&&+\  I \boX\boH_1  + J\boX \tilde{\boH}_1   + K \boW\boH_1 + L \boW\tilde{\boH}_1 +M\boX^2\boH_1 + N\boX^2\tilde{\boH}_1 +  O\boX^3  + P\boX^4 \ \nonumber 
\lab{Omega} \ea
where
\ba
&& A=(q^2-q^{-2})q^{-1}\ ,\quad B=  q^{-2}b_6 - \frac{(e_2b_2^2 -(q^2+q^{-2})(q+q^{-1})^2e_4b_2  -(q^2+q^{-2}+1)e_2b_4)}{(q^2+q^{-2})(q^2+q^{-2}+1)} ,\nonumber\\
&& C= b_4q^2\ ,\quad D=-q^{-2} \ ,\nonumber \\
&& E=   (1+q^{-2})b_7 + b_2b_6 - \frac{   (e_2b_2 - (q^2+q^{-2}) e_4)b_4  }{(q^2+q^{-2})(q^2+q^{-2}+1)}  \ ,\quad   F= (1+q^2)b_5 +b_1b_4 + e_1b_2b_4\ ,\nonumber\\
&&  G=  \frac{(b_3 + b_1b_2+e_1b_2^2)}{q-q^{-1}} +\frac{ (q^2+q^{-2}+1)q^{-1}b_4e_1}{q^2-q^{-2}}\ ,\quad  H= \frac{(b_3+ b_1b_2 + e_1b_2^2)}{q-q^{-1}} + \frac{(1+2q^{-2})qb_4e_1}{q^2-q^{-2}} \ , \nonumber\\
&& I= \frac{(b_1 + (q^2+q^{-2}+1)e_1b_2)}{q-q^{-1}}\ ,\quad  J=\frac{(b_1q^{-2} + (1+2q^{-2})e_1b_2)}{q-q^{-1}}\  , \nonumber\\
&& K=b_2/(q-q^{-1})\ , \quad  L=b_2q^{2}/(q-q^{-1})\ , \quad M= e_1 /(q-q^{-1}), \quad  N= e_1q^{-2} /(q-q^{-1})\ ,\nonumber\\
&& O=   \frac{(q^4+2q^2+4+2q^{-2}+q^{-4})e_2b_2 + (q+q^{-1})(q^2+q^{-2}) q^{-3}e_4}{(q^2+q^{-2})(q^2+q^{-2}+1)} , \quad P=e_2q^{-4}/(q^2+q^{-2})\  , \nonumber
\ea
is central in the HAW algebra.
\end{lemma}

\begin{remark} For the specialization $e_1=b_1=b_2=0$, the central element $\Omega^{HAW}$ drastically simplifies. The corresponding expression can be viewed as a `perturbation' of the central element of the Askey-Wilson algebra by the monomials $\boX^3$ and $\boX^4$ (see lemma \ref{OmAW}).
\end{remark}

\vspace{1mm}
Alternative expressions for the central element $\Omega^{HAW}$ are now considered. Observe that the defining relations (\ref{hawi1})-(\ref{hawi1}) are left invariant under the action of the automorphism (involution) $*:HAW \rightarrow HAW$ such that:
\ba
\boX^*=\boX \ ,  \quad \boW^*=\boW \quad \mbox{and} \quad (\boX\boW)^*= \boW^* \boX^* \ .
\ea
By a symmetric polynomial $Q(\boX,\boW)$ we mean a polynomial $Q(\boX,\boW)$ of two noncommuting variables $\boX,\boW$ such that
\be
Q(\boX,\boW)^*=Q(\boX,\boW) \ . \lab{sym_Q} \ee
Clearly any symmetric polynomial can be presented as a sum over elementary independent symmetric terms of a fixed degree $j=1,2,3,\dots$.

For degree $j=2$ there are three symmetric terms: $\boX^2, \boW^2$ and $\{\boX,\boW\}$.

For degree $j=3$, we have six possible terms:
\be
\boX^3, \: \boW^3, \: \{\boX^2,\boW\}, \: \{\boW^2,\boX\}, \: \boX\boW\boX, \: \boW\boX\boW\ . \lab{all_3d} \ee 

For degree $j=4$, we have ten possible terms:
\ba
&& \boX^4,\:  \boW^4, \: \{\boX^3,\boW\}, \: \{\boW^3,\boX\},\:  \boX^2\boW\boX+\boX\boW\boX^2, \: \boW^2\boX\boW+\boW\boX\boW^2, \lab{all_4d}\\
&& \{\boX^2,\boW^2\}, \: (\boX\boW)^2 +(\boW\boX)^2, \: \boX\boW^2\boX, \boW\boX^2\boW\ .\nonumber
\ea

Both defining relations \re{hawi1}-\re{hawi2} of the HAW algebra being invariant with respect to the automorphism $*$, it implies    that the central element $\Omega^{HAW}$ can be written as a polynomial $Q(\boX,\boW)$  which is symmetric with respect to $\boX$ and $\boW$. The highest total degree of such polynomial $Q(\boX,\boW)$ should be four. Moreover, we can demand that the degree of $Q(\boX,\boW)$ with respect to $\boW$ is no more than two.

What are the possible independent terms of degree three in the symmetric polynomial $Q(\boX,\boW)$? From the list \re{all_3d}, we should exclude the term $\boW^3$ as having a degree more than two in $\boW$. From the remaining terms we can exclude $\{\boX^2,\boW\}$ and $\{\boW^2,\boX\}$. Indeed, it is obvious from relations \re{hawi1}-\re{hawi2} that these terms can be expressed in terms of the monomials $\boX\boW\boX, \boW\boX\boW$ and terms of degree smaller than three. We thus see that the only independent terms of degree three in the polynomial $Q(\boX,\boW)$ are 
\be
\boX^3, \: \boX\boW\boX, \: \boW\boX\boW\ .  \lab{all_in_3d} \ee
The most difficult problem is the classification of all independent terms of degree four among all possible terms in \re{all_4d}. First of all, we should eliminate from the list \re{all_4d} the three terms having degree more than two with respect to $\boW$:  $\boW^4,\:  \{\boW^3,\boX\}, \: \boW^2\boX\boW+\boW\boX\boW^2$. Hence, the remaining terms are
\be
\boX^4, \: \{\boX^3,\boW\}, \:  \boX^2\boW\boX+\boX\boW\boX^2,  \: \{\boX^2,\boW^2\}, \: (\boX\boW)^2 +(\boW\boX)^2, \: \boX\boW^2\boX, \boW\boX^2\boW \ . \lab{all_4dr} \ee
Consider now the first relation \re{hawi1}. Multiply it first by $\boX$ from the left then by $\boX$ from the right and add both expressions. We get that the combination
\be
\{\boX^3,\boW\} + (1-q^2-q^{-2}) \left(\boX\boW\boX^2 + \boX^2\boW\boX \right) \lab{cubic_X_1} \ee
is expressible in terms of $\boX^4$ and terms of degree less than three. Hence we can exclude the term $\boX^2\boW\boX+\boX\boW\boX^2$ from the list \re{all_4dr}. In a similar manner, we can show that among four terms  $\{\boX^2,\boW^2\}, \: (\boX\boW)^2 +(\boW\boX)^2, \: \boX\boW^2\boX, \: \boW\boX^2\boW$ there are only two independent. It is convenient to take $\boZ^2$ and $\{\boX^2,\boW^2\}$ as these independent terms. Indeed, we have
\be
\boZ^2=(\boX\boW-\boW\boX)^2 = \boX\boW\boX\boW +\boW\boX\boW\boX -\boX\boW^2\boX-\boW\boX^2\boW \lab{Z^2_sym} \ee  
and hence the square of the commutator $\boZ^2$ can be taken as an independent term. Hence we arrive at the following list of four independent terms of degree four:
\be
\boX^4, \: \{\boX^3,\boW\},    \: \boZ^2, \: \{\boX^2,\boW^2\} \ . \lab{all_4df} \ee

Up to an overwall scalar factor and an additional constant term, after straightforward calculations
 we obtain from $\Omega^{HAW}$ the following symmetric form of the central element:
\ba
\Omega_{sym}^{HAW}&=& a_{40} \boZ^2 + a_{41}\boX^4 + a_{42} \{\boX^3,\boW\} + a_{43} \{\boX^2,\boW^2\}  \lab{Q_pe} \\
&& + \ a_{31} \boX^3+ a_{32}\boX\boW\boX +a_{33}\boW\boX\boW  \nonumber \\ 
&& + \ a_{21}\boX^2 + a_{22} \boW^2 + a_{23} \{\boX,\boW\}  \nonumber \\
&&+ \ a_{11}\boX + a_{12}\boW \ .\nonumber  \ea
This is the {\it minimal} symmetric form of the central element. Being not necessary for the discussion below, the explicit expressions for the coefficients are not reported here.   Such form of the central element will be convenient for studying the classical picture of the HAW algebra where commutators are replaced with Poisson brackets.

\subsection{Canonical form of the HAW algebra}
For $\rho\neq 0$,   the defining relations \re{hawi1}-\re{hawi2} of the HAW algebra can be reduced to a canonical form as we now show. Consider the following linear transformation
\be
\tilde \boX = \alpha_1 \boX + \alpha_0 \mathcal{I}, \quad \tilde \boW = \beta_1 \boW + \beta_2 \boX + \beta_0 \mathcal{I} \lab{lin_XW} \ee
for generic scalar parameters $\alpha_i, \beta_i$. It is straightforward to check that this transformation preserves the form of the commutation relations   \re{hawi1}-\re{hawi2}.  Indeed, assume for simplicity that $\alpha_1=\beta_1=1$ (the generic case can be restored trivially). Then, one can check that the new elements $\tilde \boX, \: \t \boW$ satisfy the commutation relations 
\be
[\t \boX, [ \t \boX, \t \boW]] = \rho \t \boX \t \boW \t \boX  + \t e_1 \t \boX^3+  \t b_1 \t \boX^2 + \t b_2 \{\t \boX, \t  \boW\} + \t b_3 \t \boX + \t b_4 \t \boW + \t b_5 \mathcal{I} \lab{tr_XXW} \ee 
and
\be
[\t \boW, [ \t \boW, \t \boX]] = \rho \t \boW \t \boX \t \boW +  \t e_2 \t \boX^3 + \t e_3 \t \boX \t \boW \t \boX + \t e_4 \t \boX^2 + \t b'_1 \{\t \boX, \t \boW\} + \t b_2 \t \boW^2 + \t b'_3 \t \boW + \t b_6 \t \boX + \t b_7 \mathcal{I} \lab{tr_WWX} \ee
with the restrictions
\be
\t e_3 = \t e_1(\rho+3), \quad \; \t b_1'= \t b_1 + \t e_1 \t b_2, \; \quad \t b_3' = \t b_3 + \t e_1 \t b_4 \ .\lab{JI_conds_tr} \ee
It is seen that the commutations relations \re{tr_XXW}-\re{tr_WWX} have the same form as \re{XZ_com}-\re{ZW_com} and hence they define a HAW algebra with transformed structure constants $\t e_i, \t b_i$. Explicit expressions of these new structure constants can easily be found. We present here only those which are important for our analysis:
\be
\t e_1 = e_1 - \rho \beta_2, \quad \; \t b_2 = b_2 - \rho \alpha_0, \quad \; \t b_1 = b_1 + 2\beta_2 b_2 +3\alpha_0 (\rho \beta_2-e_1) - \rho \beta_0\ . \lab{tbb} \ee 
In what follows we assume\footnote{The degenerate case $\rho=0$ will be considered elsewhere.} $\rho \ne 0$. In this generic case, it is seen from \re{tbb} that one can always choose the parameters $\alpha_0, \beta_0, \beta_2$ such that 
\be
\t e_1 = \t b_1 = \t b_2 =0\ . \lab{ebb=0} \ee 
Assuming that conditions \re{ebb=0} are satisfied, one gets what we propose as a canonical form of the HAW algebra:
\ba
\big[\t \boX, \big[ \t \boX, \t \boW\big]\big] &=& \rho \t \boX \t \boW \t \boX   + \t b_3 \t \boX + \t b_4 \t \boW + \t b_5 \mathcal{I} \ , \lab{tr_XXW_c} \\
\big[\t \boW, \big[ \t \boW, \t \boX\big]\big] &=& \rho \t \boW \t \boX \t \boW +  \t e_2 \t \boX^3 +  \t e_4 \t \boX^2   + \t b_3 \t \boW + \t b_6 \t \boX + \t b_7 \mathcal{I}\ . \lab{tr_WWX_c} 
\ea
Taking into account that $\rho=q^2+q^{-2} -2$, we can present these relations in the equivalent form
\ba
\t \boX^2 \t \boW + \t \boW \t \boX^2 - (q^2 + q^{-2}) \t \boX \t \boW \t \boX &=&   \t b_3 \t \boX + \t b_4 \t \boW + \t b_5 \mathcal{I} \ ,\lab{tr_XXW_c1} \\
 \t \boW^2 \t \boX + \t \boX \t \boW^2 - (q^2 + q^{-2}) \t \boW \t \boX \t \boW &=& \t e_2 \t \boX^3 +  \t e_4 \t \boX^2   + \t b_3 \t \boW + \t b_6 \t \boX + \t b_7 \mathcal{I} \ .\lab{tr_WWX_c2} 
\ea

The first relation \re{tr_XXW_c1} contains only linear terms in $\t \boX$ and $\t \boW$ and hence coincides with one of the defining relations for the Askey-Wilson algebra presented in \cite{Zh_AW}. The second relation \re{tr_WWX_c2} is again very close to the second relation of the Askey-Wilson algebra and differs only by two nonlinear terms $\t e_2 \t \boX^3$ and $\t e_4 \t \boX^2$. We thus see that the HAW algebra in the canonical form  \re{tr_XXW_c1}-\re{tr_WWX_c2} can be considered as a ``minimal'' perturbation of the Askey-Wilson algebra \cite{Zh_AW} where only one of the commutation relations is modified by the addition of two nonlinear terms.

\section{A relation with the Askey-Wilson algebra}
In this section, the Askey-Wilson algebra is first recalled together with its central element. Then, an explicit homomorphism from the HAW algebra to the Askey-Wilson algebra is proposed. A new presentation of the Askey-Wilson algebra follows. Special cases are discussed.

\begin{definition} The Askey-Wilson algebra  AW is generated by $X$ and $Y$ subject to the following two relations:
\ba
&&[X,[X,Y]]= \rho XYX + a_1 X^2 + a_2 \{X,Y\} + a_3 X + a_4 Y + a_5 \mathcal{I} \ , \lab{aw1} \\
&&[Y,[Y,X]]= \rho YXY + a_1\{X,Y\} + a_2 Y^2 + a_3 Y + a_6 X + a_7 \mathcal{I} \ ,\lab{AWA}
\ea
where $\rho=q^2+q^{-2}-2$ and $\{a_i\}$ are generic scalars.
\end{definition}

In the literature, the relations (\ref{aw1}) and (\ref{AWA}) are sometimes called the Askey-Wilson relations. In particular, when $q\rightarrow 1$ the parameter $\rho \rightarrow 0$ and we obtain the {\it quadratic} Racah algebra. For further convenience, let us denote:
\ba
G=[Y,X]_q \quad \mbox{and} \quad \tilde{G}=[X,Y]_q. \label{G}
\ea

Consider the AW generated by $\{X,Y\}$.
\begin{lemma}\lab{OmAW} The element
\ba
\Omega&=& (q^2-q^{-2})q^{-1}XYG + a_6q^{-2}X^2 + a_4q^2Y^2 - q^{-2}G^2  \nonumber\\
&&  + (a_7(1+q^{-2}) + a_2a_6)X +  (a_5(1+q^{2}) + a_1a_4)Y  + \frac{(a_3+a_1a_2)}{(q-q^{-1})} (G + \tilde{G}) \lab{Om}\\ 
&&+  \frac{1}{(q-q^{-1})} (a_1XG + a_1q^{-2}X\tilde{G} +a_2YG + a_2q^{2}Y\tilde{G})\ \nonumber 
\ea
is central in the AW algebra.
\end{lemma}
\begin{remark} The above lemma is a generalization of eq. (1.3) in \cite{Zh_AW}. Under the specialization $a_1=a_2=0$ and $a_4=a_6=1$, one recovers the central element of the so-called ${\mathbb Z}_3$ symmetric form of the Askey-Wilson algebra.
\end{remark}

We are now in position to give a realization of the HAW algebra in terms of the elements of the AW algebra.
\begin{pr}\label{pr1} Let $\tau_0,\tau_1,\tau_2,\tau_3,\tau_4$ be generic scalars. The map
\ba
\boX \mapsto X \ , \quad \boW \mapsto W= \tau_1 XY + \tau_2 YX + \tau_3 X + \tau_4 Y + \tau_0 \mathcal{I} \label{map} 
\ea
is an homomorphism from HAW $\rightarrow$ AW. 
\end{pr}
\begin{proof}
Firstly, we show that relation (\ref{hawi1}) holds for (\ref{map}). Using the convenient notation (\ref{G}), according to (\ref{map}), the image of the element $\boW$  in the AW algebra reads:
\ba
\boW \mapsto  \tau_0 + \tau_3 X + \tau_4 Y + \frac{(q^{-1}\tau_1 + q\tau_2)}{(q^2-q^{-2})} G + \frac{(q\tau_1 + q^{-1}\tau_2)}{(q^2-q^{-2})}\tilde{G}\ .\label{I1}
\ea 

Inserting (\ref{I1}) into (\ref{hawi1}), relation (\ref{hawi1}) reads as a combination of the monomials:
\ba
X,Y,G,\tilde{G}, XG,X\tilde{G}, X^2,X^3, X^2G,  X^2\tilde{G}\ \label{moni}
\ea
and
\ba
XY,YX, GX, \tilde{G}X, XGX, X\tilde{G}X, GX^2, \tilde{G}X^2.\label{monr}
\ea
The task is to identify the subset of irreducible monomials with respect to the defining relations of the AW algebra. According to the defining relations (\ref{aw1}) and using (\ref{G}), one has:
\ba
GX &=& q^2XG + a_1qX^2 + a_2q (XY+YX) + a_3q X + a_4q Y + a_5q \ ,\label{red1}\\
\tilde{G}X &=& q^{-2}X\tilde{G} - a_1q^{-1}X^2 - a_2q^{-1} (XY+YX) - a_3q^{-1} X - a_4q^{-1} Y - a_5q^{-1} \ ,\nonumber
\ea
Then, the second set of monomials (\ref{monr}) is reduced in terms of the first set (\ref{moni}). If the relation (\ref{hawi1}) holds, it implies that each coefficient of the irreducible monomials is vanishing. This gives a system of equations which determines uniquely $e_1,b_1,b_2,b_3,b_4,b_5$ for generic parameters $\tau_i$.
Secondly, we show (\ref{hawi2}) holds for (\ref{map}). Inserting (\ref{I1}) into (\ref{hawi2}), relation (\ref{hawi2}) reads as a combination of terms that can be reduced in terms of the two ensembles of monomials:
\ba
X,Y,G,\tilde{G}, XG,X\tilde{G}, YG,Y\tilde{G},X^2,X^3,Y^2, X^2G,  X^2\tilde{G}, YXG,XY\tilde{G}\ \label{mon2i}
\ea
and
\ba
XYG,YX\tilde{G}, XG\tilde{G},  X\tilde{G}G\ \label{mon2r}
\ea
using (\ref{red1}) and
\ba
GY &=& q^{-2}YG - a_1q^{-1}(XY+YX) - a_2q^{-1}Y^2  - a_3q^{-1} Y - a_6q^{-1} X - a_7q^{-1} \ ,\nonumber\\
\tilde{G}Y &=& q^{2}Y\tilde{G} + a_1q (XY+YX) + a_2qY^2 + a_3q Y + a_6q X + a_7q \ ,\nonumber
\ea
which follow from the defining relations (\ref{AWA}) using (\ref{G}).  
Then , the ordering of the second set of monomials (\ref{mon2r}) in terms of the set (\ref{mon2i}) requires the introduction of the central element $\Omega$ (\ref{Om}). Note that the central element can be alternatively written in terms of the monomials $YX\tilde{G}$, $\tilde{G}^2$.
If the relation (\ref{hawi2}) holds, it implies that each coefficient of the irreducible set of monomials (\ref{mon2i}) is vanishing. This gives a system of equations which determines uniquely $e_2,e_3,e_4,b'_1,b_2,b'_3,b_6,b_7$ for generic parameters $\tau_i$.
\end{proof}

To keep the exposition fluid, the explicit expressions of the structure constants $\{e_i,b_i,b'_i\}$ in terms of the structure constants $\{a_i\}$, parameters $\{\tau_i\}$ and $\Omega$ are reported in Appendix \ref{AA}. Note that for certain relations among the parameters $\{\tau_i\}$ and structure constants $\{a_i\}$, the coefficients $e_i\equiv 0$. In this case,  the relations (\ref{hawi1})-(\ref{hawi2})  of HAW algebra degenerate to those of the AW algebra. An explicit example will be discussed in the next section.

\begin{remark} In the image of the HAW algebra by the map (\ref{map}), there exists scalar parameters such that the following relations hold:
\ba
&&[Y,[Y,W]]= \rho YWY + g_1 Y^3+  c_1 Y^2 + c_2 \{Y,W\} + c_3 Y + c_4 W + c_5 \mathcal{I} \ ,\label{haw3} \\
&&[W[W,Y]]= \rho WYW + g_2 Y^3 + g_3 YWY + g_4 Y^2 + c_1\{Y,W\} + c_2 W^2 + c_3 W + c_6 Y + c_7 \mathcal{I}  \label{haw4}
\ea
\end{remark}

The existence of the homomorphism (\ref{map}) suggests that the image of the HAW algebra by (\ref{map}) might be interpreted as a subalgebra of the AW algebra generated by elements that are fixed under the action of certain automorphisms of the AW algebra.

\subsection{A different presentation}
The image of the HAW algebra by the map of Proposition \ref{pr1} is closely related with the so-called ${\mathbb Z}_3$ symmetric form of the AW algebra. In order to get this form, we assume that at least one of the coefficients $\tau_1$ and $\tau_2$ is nonzero and that
\be
\tau_1 \ne \tau_2  \lab{tt_0} \ .\ee
For the analysis that follows, it is convenient to introduce the operator $\oZ$ given by:
\be \oZ = [X,Y] \ . \lab{Z_com} \ee
Then, according to (\ref{map}) we can always present $W$ as
\be
W= \frac{\tau_1+\tau_2}{2} \{X,Y\} + \frac{\tau_1-\tau_2}{2} \oZ + \tau_3 X + \tau_4 Y + \tau_0\ . \lab{W_Z} \ee
Under condition \re{tt_0} we can express $\oZ$ in terms of $W$ and $X,Y$:
\be
\oZ=\sigma_1 W + \sigma_2 \{X,Y\} + \sigma_3 X + \sigma_4 Y + \sigma_0 \lab{Z_W} \ee
where
\be
\sigma_1 = \frac{2}{\tau_1-\tau_2}, \: \sigma_2 = \frac{\tau_1+\tau_2}{\tau_2-\tau_1}, \: \sigma_3 = \frac{2 \tau_3}{\tau_2-\tau_1}, \: \sigma_4 = \frac{2 \tau_4}{\tau_2-\tau_1}, \: \sigma_0 =  \frac{2 \tau_0}{\tau_2-\tau_1} \ .\lab{sigma_tau} \ee
It is also convenient to introduce the following polynomials in the elements $X,Y$:
\be
\Phi_1(X,Y) = [X,[X,Y]] = \rho XYX + a_1 X^2 + a_2 \{X,Y\} + a_3 X + a_4 Y + a_5 \mathcal{I} \lab{Phi_1} \ee
and
\be
\Phi_2(X,Y) =[Y,[Y,X]]= \rho YXY + a_1\{X,Y\} + a_2 Y^2 + a_3 Y + a_6 X + a_7 \mathcal{I}, \lab{Phi_2} \ee
where we have expoited the AW relations \re{aw1} and \re{AWA}. Let us here record that the elementary operator identities
\be
\{X,\{X,Y \}\} = [X,[X,Y]] + 4 XYX = \Phi_1(X,Y) + 4 XYX \lab{CMAM_1} \ee
and
\be
\{Y,\{Y,X \}\} = [Y,[Y,X]] + 4 YXY = \Phi_2(X,Y) + 4 YXY \lab{CMAM_2} \ee
which will be useful.
Consider now the commutation relation $[X,Y]=\oZ$. Expressing $\oZ$ via \re{Z_W} we have 
\be
[X,Y] = \sigma_1 W + \sigma_2 \{X,Y\} + \sigma_3 X + \sigma_4 Y + \sigma_0 \mathcal{I} . \lab{XYZ} \ee
Consider the commutator $[W,X]$. From \re{W_Z}, we have
\ba
&&[W,X]= -\frac{\tau_1+\tau_2}{2} \{X,\oZ\} + \frac{\tau_1-\tau_2}{2} [\oZ,X] - \tau_4 \oZ = \lab{WX1}  \\
&& -\frac{\tau_1+\tau_2}{2} \{X,\oZ\}- \tau_4 \oZ - \frac{\tau_1-\tau_2}{2} \Phi_1(X,Y) \ ,\nonumber\ea
which can be simplified by substituting $\oZ$ from \re{Z_W}. It follows that
\ba
[W,X] &=&  -\frac{\tau_1+\tau_2}{2} \left(\sigma_1 \{W,X\} + \sigma_2 \{X,\{X,Y\}\} + 2 \sigma_3 X^2 + \sigma_4 \{X,Y\} + 2 \sigma_0 X \right) \lab{WX-1}  \\
&& \qquad - \tau_4 \left( \sigma_1 W + \sigma_2 \{X,Y\} + \sigma_3 X + \sigma_4 Y + \sigma_0 \mathcal{I} \right) -   \frac{\tau_1-\tau_2}{2} \Phi_1(X,Y).   \nonumber \ea
Finally, we can use formula \re{CMAM_1} in order to simplify the double anticommutator in \re{WX-1}:
\ba
[W,X] &=&  -\frac{\tau_1+\tau_2}{2} \left(\sigma_1 \{W,X\}  + 4 \sigma_2  XYX  + 2 \sigma_3 X^2 + \sigma_4 \{X,Y\} + 2 \sigma_0 X \right)    \lab{WX-2} \\
&& - \ \tau_4 \left( \sigma_1 W + \sigma_2 \{X,Y\} + \sigma_3 X + \sigma_4 Y + \sigma_0 \mathcal{I} \right) - \left(  \sigma_2 \, \frac{\tau_1+\tau_2}{2} + \frac{\tau_1-\tau_2}{2} \right)\Phi_1(X,Y). \nonumber \ea
Combining similar terms, we arrive at the relation
\be
[W,X] = \sigma_2 \{W,X\} + \kappa \: XYX + \mu_1 X^2 + \mu_2 \{X,Y\} + \mu_3 X + \mu_4 Y + \mu_5 \mathcal{I}. \lab{WX-3} \ee 
In a similar manner we can express the last commutator $[Y,W]$ like this:
\be
[Y,W]= \sigma_2 \{W,Y\} + \kappa \: YXY + \mu_1 \{X,Y\} + \mu_2 Y^2 + \mu_3 Y + \mu_6 X + \mu_7 \mathcal{I} \lab{WY-3} \ee 
where
\ba
&&\kappa = 2\,{\frac { \left( {q}^{2}\tau_{{1}}+\tau_{{2}} \right)  \left( \tau_{{1}}+{q}^{2}\tau_{{2}} \right) }{ \left( \tau_{{1}}-\tau_{{2}}
 \right) {q}^{2}}}, \quad \lab{kappa} \ea
and
\ba
&&\mu_1 = 2\,{\frac {\tau_{{2}}\tau_{{1}}a_{{1}} +\tau_{{2}}\tau_{{3}}+\tau_{{1}} \tau_{{3}}}{\tau_{{1}}-\tau_{{2}}}}, \quad 
\mu_2 = 2\,{\frac {\tau_{{2}}\tau_{{1}}a_{{2}}+ \tau_{{2}}\tau_{{4}}+\tau_{{1}}\tau_{{4}}}{\tau_{{1}}-\tau_{{2}}}}, \lab{mumu} \\
&&\mu_3=2\,{\frac {\tau_{{2}}\tau_{{0}}+\tau_{{1}}\tau_{{0}}+\tau_{{4}}\tau_{{3}}+\tau_{{2}}\tau_{{1}}a_{{3}}}{\tau_{{1}}-\tau_{{2}}}}, \quad    
\mu_4 = 2\,{\frac {{\tau_{{4}}}^{2}+\tau_{{2}}\tau_{{1}}a_{{4}}}{\tau_{{1}}-\tau_{{2}}}}, \quad \mu_5= 2\,{\frac {\tau_{{2}}\tau_{{1}}a_{{5}}+\tau_{{4}}\tau_{{0}}}{\tau_{{1}
}-\tau_{{2}}}},
 \nonumber \\
&&\mu_6= 2\,{\frac {{\tau_{{3}}}^{2}+\tau_{{2}}\tau_{{1}}a_{{6}}}{\tau_{{1}}-\tau_{{2}}}}, \quad \mu_7= 2\,{\frac {\tau_{{3}}\tau_{{0}}+\tau_{{2}}\tau_{{1}}a_{{7}}}{\tau_{{1}}-\tau_{{2}}}}. \nonumber \ea
It follows that \re{XYZ}, \re{WX-3} and \re{WY-3} can be presented in a slightly different form:
\ba
&&XY- p YX=  \tau_1^{-1}\left(  W - \tau_3 X - \tau_4 Y - \tau_0 \mathcal{I} \right)\ , \nonumber \\
&&WX-pXW = \t \kappa \: XYX + \t \mu_1 X^2 + \t \mu_2 \{X,Y\} + \t \mu_3 X + \t \mu_4 Y + \t \mu_5 \mathcal{I} \ , \lab{XYW_p} \\
&&YW-pWY = \t \kappa \: YXY + \t \mu_1 \{X,Y\} + \t \mu_2 Y^2 + \t \mu_3 Y + \t \mu_6 X + \t \mu_7 \mathcal{I} \ , \nonumber \ea 
where
\be
p= -\frac{\tau_2}{\tau_1}, \quad \t \kappa = \frac{\kappa}{1-\sigma_2}, \quad \t \mu_i = \frac{\mu_i}{1-\sigma_2} , \quad i=1,2,\dots, 7. \lab{t_km} \ee


By Proposition \ref{pr1},  the commutation relations \re{XYZ}, \re{WX-3} and \re{WY-3} - or equivalently \re{XYW_p} - provide a presentation of the AW algebra with basic generators $X,Y,W$ that differs from the original one \cite{Zh_AW}. Let us point out that the form of the defining relations in this presentation  looks very close to the structure of the relations occuring in the three dimensional Calabi-Yau and Sklyanin algebras \cite{BT}, \cite{EG}. In this presentation, let us observe that the r.h.s. of the relations \re{XYW_p} contain (apart from linear terms) only {\it symmetric} nonlinear combinations of the generators
\ba
XYX, \: YXY, \: \{X,Y\}, \: \{W,X\}, \: \{W,Y\}. \nonumber
\ea
 The highest degree of polynomials in $X,Y$ in the r.h.s. of these relations is three with the cubic terms in \re{XYW_p}  being $\kappa XYX$ and $\kappa YXY$. The condition $\kappa=0$, that has these terms vanishing is equivalent to one of two conditions: 
\be
\tau_2 = -q^2 \tau_1 \quad  \mbox{or}  \quad \tau_1=-q^2 \tau_2\ . \lab{tau_q_cond} \ee
Interestingly, it is well known that the conditions \re{tau_q_cond} convert the AW algebra with relation \re{aw1} and \re{AWA} to the ${\mathbb Z}_3$ symmetric form
\be
[X,Y]_q = W + \omega_3 \mathcal{I}, \qquad [Y,W]_q=X + \omega_1 \mathcal{I}, \qquad [W,X]_q = Y + \omega_2 \mathcal{I}, \lab{AW_Z3} \ee
where an appropriate affine transformation of generators 
\be
X \to \alpha_1 X +\beta_1 \mathcal{I}, \qquad Y \to \alpha_2 Y +\beta_2 \mathcal{I} \lab{affine_XY} \ee
is applied (it is always possible to choose the parameters $\alpha_i, \beta_i$ such that all nonlinear terms vanish in the r.h.s. of \re{AW_Z3}). The structure of the algebra \re{XYW_p} is very close to \re{AW_Z3} with the identification $p=q^2$ or $p=q^{-2}$ depending on the choice in \re{tau_q_cond}. We thus see that algebra \re{XYW_p} is a natural generalization of the cyclic (${\mathbb Z}_3$ symmetric) form of the $AW$ algebra \re{AW_Z3} which involves the image of the HAW operator $W$ in the AW algebra \re{W_Z}.

It is also interesting to note that the only exception to the above scheme occurs when either $\tau_1=\tau_2$ or $\tau_1=\tau_2=0$. In the first case, the element $W$ can be presented as
\be
W = \tau_1 \{X,Y\} + \tau_3 X + \tau_4 Y + \tau_0 \mathcal{I} .\lab{W_bt} \ee
Element $W$ of this type occur in band-time limiting problems \cite{GVZ_band}. Indeed, the commuting operators W in these signal processing  problems (see  \cite{GVZ_band} for details) must be symmetric in both eigenbases of the bispectral operators $X$ and $Y$ and this requires $W$ to be of the form \re{W_bt}.

In the second case (i.e. when $\tau_1=\tau_2=0$), the element $W$ is a simple linear combination of  the AW generators
\be
W = \tau_3 X + \tau_4 Y + \tau_0 \mathcal{I} \lab{W_lin} \ee
in which case  there is no analog of the  ${\mathbb Z}_3$ symmetric form. \vspace{2mm}

To conclude this section, let us mention that the Askey-Wilson algebra can be embedded into the $U_q(sl_2)$ algebra. Namely, the elements $X,Y$ are written in terms of twisted primitive elements of $U_q(sl_2)$. For instance, see \cite[Section 2.4]{BMVZ}. Using Proposition  \ref{pr1}, an embedding of HAW into $U_q(sl_2)$ follows. Also, using Proposition \ref{pr1} irreducible  finite dimensional representations of the HAW algebra can be constructed within the theory of Leonard pairs \cite{Ter}. In the next section, an infinite dimensional representation of the HAW algebra is given.

\section{The Heun-Askey-Wilson operator and bispectrality}
 In this section we construct an infinite dimensional representation of the Heun-Askey-Wilson algebra. Firstly, we build the most general $q$-deformed analog of the Heun operator. Secondly, in keeping with the last section, this operator is shown to be a bilinear combination of bispectral $q-$difference operators of Askey-Wilson type.  By Proposition {\ref{pr1}}, this operator together with the $q-$difference operators associated with $X,Y$ satisfy the defining relations of the algebra   HAW. The bispectrality properties of the Heun operator of Askey-Wilson type are also described.  

\subsection{The Heun-Askey-Wilson operator} Consider the vector space of the polynomials in the variable $x = z+z^{-1}$.
Let
 \be
T^+ f(z) = f(q^2z), \qquad T^- f(z) = f(z/q^2) \lab{TPM_def} 
\ee
be q-shift operators, where the functions $f(z)$ are assumed to be Laurent polynomials in $z$.  We will deal with second-order $q$-difference operators of the type
\be
{\mathcal W}= A_1(z) T^+ + A_2(z) T^- + A_0(z) \mathcal{I} \ . \lab{gen_W_def} 
\ee
The operator ${\mathcal W}$ can be considered as a generalization of the Askey-Wilson operator which entails the difference equation for the Askey-Wilson polynomials. We start with arbitrary functions $A_i(z)$ and ask when is the operator ${\mathcal W}$ mapping polynomials $\pi_n(x)$ of degree $n$ onto polynomials $\rho_{n+1}(x)$ of degree $n+1$. We thus demand that
\be
{\mathcal W} \pi_n(x) = \rho_{n+1}(x) , \quad n=0,1,2,\dots \lab{W_def} \ee
First, we find the necessary conditions for the operator ${\mathcal W}$ to satisfy the property (\ref{W_def}). Taking $\pi_n(x)=x^n, \;n=0,1,2, \dots$ For $n = 0,1,2$,  (\ref{W_def}) yields 
\be
A_1(z) + A_2(z) + A_0(z) = p_1(x) \lab{0_cond}
\ee
with $p_1(x)$ a polynomial of degree one. For $n=1$ and $n=2$, we obtain a linear system of equations with 
solution 
\be
A_1(z) = \frac{Q(z)}{z(1-z^2)(1-q^2z^2)}, \quad A_2(z) = A_1(1/z) , \quad A_0(z) = p_1(z) -A_1(z) - A_2(z), \lab{A012}
 \ee
where
\be
Q(z) =\sum_{k=0}^6 r_k z^k \lab{Q_z} 
\ee
is an arbitrary polynomial of degree six.

\begin{remark} When $r_0=r_6=0$, the polynomial $Q(z)$ becomes 
\be
Q=z\mathcal{P}(z), \lab{spec_Q} \ee
where $\mathcal{P}(z)$ is an {\it arbitrary} polynomial of degree four. Up to a common factor, one can present it as
\be
\mathcal{P}(z) =(1-\xi_1 z)(1-\xi_2 z)(1-\xi_3 z)(1-\xi_4 z). \lab{AW_P} \ee
Then the coefficients $A_1(z)$ and $A_2(z)$ become
\be
A_1^{(0)}(z) = \frac{\mathcal{P}(z)}{(1-z^2)(1-q^2z^2)}, \quad A_2^{(0)}(z) =  A_1^{(0)}(1/z). \lab{A12_AW} \ee

If additionally, one puts $p_1(x)=0$, one obtains the well known difference operator for the Askey-Wilson polynomials:
\be
{\mathcal Y} =  A_1^{(0)}(z)\left(T^+ - \mathcal{I} \right) + A_2^{(0)}(z)\left(T^- - \mathcal{I} \right). \lab{Y_AW} \ee
Thus, the Askey-Wilson operator ${\mathcal Y} $ is a special case of the Heun-AW operator ${\mathcal W} $ with the additional restrictions \re{spec_Q} and $p_1(x)=0$.
\end{remark}

So far, we have only obtained the necessary conditions for the operator ${\mathcal W} $ to satisfy the property (\ref{W_def}). It is easy to show that these conditions are also sufficient. We thus have 
\begin{pr}\label{pr4}
The most general HAW operator of the form \re{gen_W_def} which is mapping polynomials of degree $n$ in $x$ into polynomials of degree $n+1$ must have its coefficients  $A_0(z), A_1(z), A_2(z)$ given by the expressions \re{A012}
 with an arbitrary polynomial $Q(z)$ of degree six. In the special case when $Q(z)=z \mathcal{P}(z), \; p_1(x)=0$ (with $\mathcal{P}(z)$ a polynomial of degree four) one obtains generic Askey-Wilson operator ${\mathcal Y} $.
\end{pr}
Another important characterization of the HAW operator ${\mathcal W}$ is given by 
\begin{pr}
The HAW operator ${\mathcal W}$ defined by Proposition \ref{pr4} coincides with the algebraic Heun operator associated to the Askey-Wilson polynomials, that is
\be
{\mathcal W} = \tau_1 {\mathcal X}  {\mathcal Y}  + \tau_2 {\mathcal Y} {\mathcal X}  + \tau_3{\mathcal X}  + \tau_4 {\mathcal Y}  + \tau_0 \mathcal{I}, \lab{alg_W} \ee
with ${\mathcal Y}$ given by (\ref{Y_AW}) and the operator ${\mathcal X} $ being multiplication by $x$:
\be
{\mathcal X}  = z+z^{-1} \ .  \lab{X_def} \ee
\end{pr}
The proof of this proposition is straighforward and relies on the fact that the HAW operator defined by \re{gen_W_def} and \re{A012} is recovered from the substitution into  (\ref{alg_W}) of ${\mathcal X} $ and ${\mathcal Y}$  as given respectively by (\ref{X_def})  and  (\ref{Y_AW}) .

In more details, let us consider the operator \re{alg_W} with arbitrary parameters $\tau_i$. A simple calculation shows that the operator \re{alg_W} has expression \re{gen_W_def}, where
\be
A_1(z) =  \left[ (\tau_1 + q^2 \tau_2) z + (\tau_1 + q^{-2} \tau_2) z^{-1} + \tau_4\right]  A_1^{(0)}(z)  , \quad A_2(z) = A_1(1/z) \lab{A12_0} \ee
and 
\be
A_1(z) + A_2(z) +A_0(z) = p_1(x) = \kappa_1 x + \kappa_0 \lab{sum_A_p} \ee
with the coefficients $\kappa_1, \kappa_0$ expressed in terms of $\xi_i$ and $\tau_i$. Expressions \re{A12_0} and \re{sum_A_p} show that for generic choice of the parameters $\tau_i$, the algebraic Heun  operator \re{alg_W} coincides with the operator \re{gen_W_def}. The inverse statement follows easily from the same considerations.

By Proposition {\ref{pr1}}, the HAW operator satisfies the defining relations of the Heun-Askey-Wilson algebra.   The explicit expressions of the structure constants of this algebra in terms of the parameters $\{\xi_i\}$ and $\{\tau_i\}$  of this realization are rather complicated. Nevertheless, one can consider the special case when all extra parameters vanish, i.e.
\be
e_i=g_i=0, \; i=1,2,3,4. \lab{eg=0} \ee
It happens only in two cases (in fact these cases are almost the same). Either
\be
\tau_2=-q^2 \tau_1, \; \tau_4 =0, \: \tau_3= \tau_1 (q^{-2}-1) \left(q^2+\xi_1 \xi_2 \xi_3 \xi_4\right) \lab{case_i} 
\ee
or
\be
\tau_1=-q^2 \tau_2, \; \tau_4 =0, \: \tau_3= \tau_2 (q^{-2}-1) \left(q^2+\xi_1 \xi_2 \xi_3 \xi_4\right). \lab{case_ii} \ee
\vspace{1mm}

It is seen, that in these cases all three pairs $({\mathcal X},{\mathcal Y}), ({\mathcal X},{\mathcal W})$ and $({\mathcal Y},{\mathcal W})$ satisfy the defining relations of the AW algebra.

Note also that for the above realization,  the central element $\Omega^{HAW}$ given by \re{Omega}  (this element plays the role of the Casimir operator of the Heun-Askey-Wilson algebra) reduces to the identity operator multiplied by a constant. This situation is similar to the case of the ordinary Askey-Wilson algebra  \cite{Zh_AW}, where the Casimir operator reduces to the identity operator as well.

In contrast to the tridiagonalization of the ordinary hypergeometric operator, or of the little q-Jacobi operator,
the algebraic Heun operator construct associated to the Askey-Wilson bispectral operator entails a rather rigid structure for the choices  \re{case_i} or \re{case_ii}. Indeed, only two parameters are arbitrary: $\tau_1$ and $\tau_0$ and this arbitraryness is equivalent to an affine transformation ${\mathcal  W}\to \alpha {\mathcal W} + \beta \mathcal{I}$. This means that the operator ${\mathcal W}$ does not contain any nontrivial parameter. This new operator ${\mathcal W}$ is nevertheless interesting because it gives a special nontrivial case when the HAW operator  leads to functions of the hypergeometric type.

\subsection{Bispectrality} The Heun-Askey-Wilson operator ${\mathcal W}$ possesses
properties relative to its action on specific bases that parallel those of the Heun-Hahn and $q$-Heun operators \cite{VZ_HH}, \cite{BVZ}, as we shall now indicate. Consider the Askey-Wilson polynomials $P_n(x)$ which are eigenfunctions of the operator ${\mathcal Y}$:
\be
{\mathcal Y} P_n(x) = \lambda_n P_n(x). \lab{YPP} \ee
The HAW operator ${\mathcal W}$ is tridiagonal in the basis $P_n(x)$:
\be
{\mathcal W} P_n(x) = r_{n+1}^{(1)} P_{n+1}(x) + r_{n+1}^{(2)} P_{n}(x) + r_{n+1}^{(3)} P_{n-1}(x) .\lab{W_P_tri} \ee
The operator ${\mathcal X}$ is tridiagonal in this basis:
\be
{\mathcal X} P_n(x) = P_{n+1}(x) + b_n P_n(x) + u_n P_{n-1}(x) \lab{X_P_tri} \ee
because \re{X_P_tri} is nothing else than the three-term recurrence relation for the AW polynomials. The operator ${\mathcal Y}$ is diagonal \re{YPP} and hence the operator ${\mathcal W}$ given by \re{alg_W} has the tridiagonal action \re{W_P_tri}. The coefficients $ r_{n+1}^{(2)}$ are easily expressed in terms of the tridiagonalization parameters and of the recurrence coefficients   $u_n,b_n$  and the eigenvalues $\lambda_n$ associated to the AW polynomials.

Besides the Askey-Wilson polynomial  basis, let us now consider also the following (nonmonic) basis:
\be
\vphi_n(x)= (\xi_1 z ; q^2)_n(\xi_1/z;q^2)_n
\lab{AW_basis} \ee
It is well known \cite{KM}, \cite{VZ_Newton} that the operator ${\mathcal Y}$ is bidiagonal with respect to the basis $\vphi_n(x)$: 
\be
{\mathcal Y} \vphi_n(x) = \lambda_n \vphi_n(x) + \nu_n \vphi_{n-1}(x), \lab{Y_phi} \ee
where 
\begin{eqnarray}
&& \lambda_n = \dfrac{(1-q^{2n})(1-\xi_1\xi_2\xi_3\xi_4 q^{2n-2})}{q^{2n}}\\
&& \nu_n = \dfrac{(q^{2n}-1)(1-\xi_1\xi_2 q^{2n-2})(1-\xi_1\xi_3q^{2n-2})(1-\xi_1\xi_4q^{2n-2})}{q^{2n}}
\end{eqnarray}
while the operator ${\mathcal X}$ is obviously bidiagonal in this basis
\be
{\mathcal X} \vphi_n(x) = \mu_n \vphi_{n+1}(x) + \rho_n \vphi_n(x) \lab{X_phi} ,
\ee
where 
\be   \mu_n = -\xi_1^{-1} q^{-2n} , \quad \rho_n = \xi_1 q^{2n} + \xi_1^{-1} q^{-2n} \ .
\ee
It follows that ${\mathcal W}$  is tridiagonal in this basis also.

\vspace{2mm}

\begin{remark}
\setcounter{equation}{0}
This striking similarity of the HAW algebra with the Askey-Wilson algebra leads to some important conclusions concerning possible discrete spectrum of the operator $\t X$ in the canonical form of the HAW  \re{tr_XXW_c}-\re{tr_WWX_c}. Indeed, assume that there exists a discrete basis $\vphi_n$ such that $\t X$ is diagonal and $\t W$ is tridiagonal in this basis:
\be
\t X \vphi_n= \lambda_n \vphi_n, \quad \t W \vphi_n = a_{n+1} \vphi_{n+1} + g_n \vphi_n + a_n \vphi_{n-1} \lab{XW_phi} \ee    
Then from commutation relations  \re{tr_XXW_c1}-\re{tr_WWX_c2} we conclude that $\lambda_n$ satisfy the relations
\be
\lambda_n^2 + \lambda_{n+1}^2 - \left(q^2 + q^{-2}\right) \lambda_n \lambda_{n+1} = \t b_4 \lab{eql_1} \ee
and
\be
\lambda_{n+1} + \lambda_{n-1} - \left(q^2 + q^{-2}\right) \lambda_n =0 \ .\lab{eql_2} \ee
From \re{eql_1}-\re{eql_2} we derive the explicit expression for $\lambda_n$:
\be
\lambda_n = \xi_1 q^{2n} + \xi_2 q^{-2n}, \lab{lambda_AW} \ee 
where the parameters $\xi_1, \xi_2$ satisfy the condition
\be
\xi_1 \xi_2 = -\t b_4 \: \left(q^2 - q^{-2}\right)^{-2} \lab{xi_xi_b} \ee
We thus see that the operator $\t X$ has the same ``classical" (i.e. Askey-Wilson) spectrum \re{lambda_AW} as generators of the Askey-Wilson algebra. In contrast, the operator $\t W$ is "nonclassical". A detailed study of the representation  \re{XW_phi} and  (\ref{X_phi}) will be deferred to another publication.
\end{remark}

\vspace{2cm}

\section{Conclusion}
\setcounter{equation}{0}
Let us sum up now. This paper has been concerned with the Heun type extension of the algebra and operators associated to the important class of Askey-Wilson orthogonal polynomials. Anticipating on the concrete realization, we have first introduced abstractly the Heun-Askey-Wilson (HAW) algebra in terms of generators and relations and discussed its universal structure in terms of symmetric polynomials in non-commuting variables. The expression for the central element of this algebra was provided and a canonical presentation analogous to the one of the Askey-Wilson algebra has been obtained. We have then exhibited explicitly an embedding of the HAW algebra into the Askey-Wilson one. This has allowed to obtain yet another presentation related to the ${\mathbb Z}_3$ symmetric form of the AW algebra.

We have then focused on the realization of the HAW algebra that results when one introduces the Heun operator of Askey-Wilson type that generalizes to the Askey-Wilson lattice the standard Heun operator in the continuum. Quite strikingly, it was seen that there is coincidence between the Algebraic Heun Operator construction applied to the AW bispectral operators and the determination of the most general second order q-difference operator that raises by one the degree of polynomials in $x = z+z^{-1}$. This has the immediate implication that the Heun-Askey-Wilson operator is tridiagonal on the basis of Askey-Wilson polynomials and that it has in fact the same property in the AW basis.

These advances raise a number of questions. We already mentioned that it would be interesting to investigate if the HAW algebra could identified as a subalgebra of the AW algebra generated by elements that are fixed under the action of certain automorphisms of the AW algebra. We closed the last section by bringing up representation theory issues that we plan to explore. It would evidently be of great interest to use the present study as a stepping stone to undertake in the same spirit the study of multivariable Heun operators and to look at their applications in the realm of integrable models. We hope to report on some of these questions in the near future.

\vspace{3mm}

{\Large\bf Acknowledgments}

PB and AZ would wish to acknowledge the hospitality of the CRM during the course of this work. AZ is grateful for V.Roubtsov for fruitful discussions. PB thanks N. Cramp\'e for help with MAPLE. PB is supported by the CNRS.  The research
of L.V. is funded in part by a discovery grant from the Natural Sciences and Engineering Research
Council (NSERC) of Canada. The work of A.Z. is supported by the National
Science Foundation of China (Grant No.11771015).

\vspace{15mm}

\begin{appendix}

\section{Structure constants}\label{AA}
We provide in this appendix the explicit formulas for the structure constants of the HAW algebra in terms of those of the AW algebra and of the parameters $\{\tau_i\}$ that define the embedding of the former algebra in the latter.
For generic parameters $\tau_0,\tau_1,\tau_2,\tau_3,\tau_4$, the solution is unique, given by:
\ba
b_1 &=& (\tau_1+\tau_2)a_3 - (q-q^{-1})^2\tau_0 + \tau_4a_1 -2\tau_3a_2 \ ,\nonumber\\
b_2 &=& a_2\ ,\nonumber\\
b_3 &=& (\tau_1 + \tau_2)a_5 +  - \tau_3a_4 + \tau_4a_3  -2\tau_0a_2 \ ,\nonumber\\
b_4 &=& a_4 \ ,\nonumber\\
b_5 &=& -\tau_0a_4 + \tau_4 a_5 \ ,\nonumber\\
e_1 &=& (\tau_1+\tau_2)a_1 - (q-q^{-1})^2 \tau_3 \ .\nonumber
\ea
and
\ba
b'_1 &=&     (( {\tau_1}+{\tau_2})a_2 + \tau_4)a_1 -(q^2+q^{-2})\tau_3a_2 + ({\tau_1}+{\tau_2})a_3 -(q-q^{-1})^2 \tau_0
\ ,\nonumber\\
b'_3 &=&  ( {\tau_1}+{\tau_2})a_4a_1-
2 \tau_0a_2+\tau_4a_3+ (1-q^2-q^{-2}) \tau_3 a_4+\tau_1a_5+a_5\tau_2
 \ ,\nonumber
\ea
\ba
b_6 &=&   \tau_1\tau_2 (-a_1a_2a_3 +a^2_2a_6 -a_3^2 + a_4a_6 -(q^2+q^{-2})a_1a_5 + (q+q^{-1})^2a_2a_7)    \nonumber\\
&& + (\tau_1 + \tau_2)(-2\tau_0 a_1a_2 - \tau_3a_1a_4 + \tau_4 a_2a_6 - 2\tau_0a_3 - (q^2+q^{-2})\tau_3a_5 + (q+q^{-1})^2\tau_4a_7)\nonumber\\
&& -2\tau_0\tau_4a_1 +  \tau_4^2a_6 + (q^2+q^{-2})(2\tau_0\tau_3a_2 - \tau_3\tau_4 a_3) + (q^2+q^{-2}-1)\tau_3^2a_4 + (q-q{-1})^2\tau_0^2\nonumber\\
&& -  (\tau_1 + q^{2}\tau_2)(\tau_1 + q^{-2}\tau_2)\Omega \ ,\nonumber\\
b_7 &=& \tau_1\tau_2 (-a_1a_2a_5 -\Omega a_2 -a_3a_5 + a_4a_7) - (\tau_1+\tau_2)(\tau_0a_1a_4 + \tau_0a_5 + \Omega\tau_4) \nonumber\\
&& + \tau_0^2a_2 -\tau_0\tau_4a_3 + \tau_4^2a_7 + (q^2+q^{-1}-1)\tau_3(\tau_0a_4 - \tau_4 a_5)\  ,\nonumber\\
e_2 &=&  (q^2+q^{-2})\left( -\tau_1\tau_2 a_1^2 -2\tau_3(\tau_1+\tau_2)a_1 + (\tau_1 + q^{2}\tau_2)(\tau_1 + q^{-2}\tau_2)a_6 -(q-q^{-1})^2\tau_3^2 \right)                           \ ,\nonumber\\
e_3 &=&  (q^2+q^{-2}+1)\left((\tau_1+\tau_2)a_1 -(q-q^{-1})^2\tau_3\right) \ ,\nonumber\\
e_4 &=&  -\tau_1\tau_2a_2a_1^2    -2(\tau_1+\tau_2)\tau_3 a_1a_2 - (1+q^2+q^{-2})\left(\tau_1\tau_2a_1a_3 +(\tau_1+\tau_2)\tau_0a_1 + \tau_3\tau_4a_1    \right) \nonumber\\
&& (\tau_1^2 + (1+2q^2+2q^{-2})\tau_1\tau_2 + \tau_2^2)a_2a_6 +  (-1+2q^2+2q^{-2})\tau_3^2a_2 \nonumber\\
&& +   (1+q^2+q^{-2})\left((\tau_1+\tau_2)(-\tau_3a_3 + \tau_4a_6) +  (\tau_1 + q^{2}\tau_2)(\tau_1 + q^{-2}\tau_2)a_7   + (q-q^{-1})^2\tau_3\tau_0 \right)
\ .\nonumber
\ea

\end{appendix}

\vspace{2mm}

\bb{99}

\bi{BMVZ} P. Baseilhac, X. Martin, L.Vinet and A. Zhedanov, {\it Little  and big q-Jacobi polynomials and  the Askey-Wilson algebra}, arXiv:1806.02656.

\bi{BVZ} P. Baseilhac, L. Vinet and A. Zhedanov, {\it The q-Heun operator of the big q-Jacobi type and the q-Heun algebra}, arxiv:1808.06695.

\bi{BT} R. Berger and R. Taillefer, {\it Poincare-Birkhoff-Witt deformations of Calabi-
Yau algebras}. J. Noncommut. Geom., {\bf 1}(2): 241--270, 2007. https://arxiv.org/abs/math/0610112

\bi{Diejen} J.F. van Diejen, {\it Integrability of difference Calogero-Moser systems}, J. Math. Phys. {\bf 35} (1994), 2983--3004.

\bi{EG} P. Etingof and V. Ginzburg, {\it Noncommutative del Pezzo surfaces and Calabi-Yau algebras}  https://arxiv.org/abs/0709.3593

\bi{GIVZ} V. Genest, M.E.H. Ismail, L. Vinet, A. Zhedanov {\it Tridiagonalization of the hypergeometric operator and the Racah-Wilson algebra}, arXiv:1506.07803.


\bi{GLZ_Annals} Ya. A. Granovskii, I.M. Lutzenko, and A. Zhedanov, {\it Mutual integrability, quadratic algebras,
and dynamical symmetry}. Ann. Phys. {\bf 217} (1992),  1--20.

\bi{GVZ_Heun} F.A. Gr\"unbaum, L. Vinet and A. Zhedanov, {Tridiagonalization and the Heun equation}, J.Math.Phys. {\bf 58}, 031703 (2017), arXiv:1602.04840.

\bi{GVZ_band} F.A. Gr\"unbaum, L. Vinet and A. Zhedanov, {Algebraic Heun operator and band-time limiting},  arXiv:1711.07862.


\bi{IK1} M. E. H. Ismail and E. Koelink. {\it Spectral analysis of certain Schr\"odinger operators}. SIGMA
{\bf 8} (2012) 61-79, arXiv:1205.0821.

\bi{IK2} M. E. H. Ismail and E. Koelink. {\it The J-matrix method}.
Adv. Appl. Math. {\bf 56}  (2011) 379-395, arXiv:0810.4558.

\bi{KM} E.G. Kalnins and W. Miller, {\it Symmetry techniques for q-series: Askey-Wilson polynomials}, Rocky Mount. J.Math. {\bf 19} (1989), 223-230.

\bibitem{KLS} R. Koekoek, P.A. Lesky, and R.F. Swarttouw. {\it Hypergeometric orthogonal polynomials and their q-analogues}. Springer, 1-st edition, 2010.




\bi{Ruijsenaars} Ruijsenaars S.N.M., {\it Integrable $BC_N$ analytic difference operators: hidden parameter symmetries and eigenfunctions},
in New Trends in Integrability and Partial Solvability, NATO Sci. Ser. II Math. Phys. Chem.,
Vol. 132, Kluwer Acad. Publ., Dordrecht, 2004, 217--261.

\bi{Takemura} K. Takemura {\it On $q$-deformations of Heun equation}, arxiv:1712.09564.

\bi{Ter} P.Terwilliger, {\it Two linear transformations each tridiagonal with respect to an eigenbasis of the other}, Lin.Alg.Appl. {\bf 330} (2001), 149--203, arXiv:math/0406555.


\bi{VZ_Newton} L. Vinet and A. Zhedanov, {\it Hypergeometric Orthogonal Polynomials with respect to Newtonian Bases}, SIGMA {\bf 12} (2016), 048, 14 pages.

\bi{VZ_HH} L. Vinet and A. Zhedanov, {\it The Heun operator of Hahn type}, arxiv:1808.00153.

\bi{Zh_AW} A. S. Zhedanov, {\it "Hidden symmetry" of Askey-Wilson polynomials}, Theoret. and Math. Phys. {\bf 89}
(1991), 1146-1157.



\end{thebibliography}

\end{document}